\documentclass[11pt,letterpaper]{article}
\pdfoutput=1

\usepackage{epsfig}
\usepackage{graphicx}
\usepackage{amsthm}
\usepackage{amssymb}
\usepackage{amsmath}
\usepackage{amsfonts}
\usepackage{mathrsfs}
\usepackage{algorithm}
\usepackage{algorithmic,subfigure,multirow}

\usepackage{fullpage}

\usepackage{color}

\newtheorem{lemma}{Lemma}
\newtheorem{theorem}{Theorem}

\newtheorem{definition}{Definition}

\newcommand{\commentout}[1]{}

\newcommand{\defer}[1]{}

\begin{document}

\title{SeeSite: Efficiently Finding Co-occurring Splice Sites and Exon Splicing Enhancers}
\author{
			Christine Lo\thanks{These authors contributed equally to this work.}~\thanks{Department of Computer Science and Engineering, University of California, San Diego}\
\and 	Boyko Kakaradov\footnotemark[1]~\thanks{ Bioinformatics Program, University of California, San Diego, USA}
\and		Daniel Lokshtanov\footnotemark[2]
\and 	Christina Boucher\footnotemark[2]
}

\maketitle

\begin{abstract}

The problem of identifying splice sites consists of two sub-problems: finding their boundaries, and characterizing their sequence markers. Other splicing elements---including, enhancers and silencers---that occur in the intronic and exonic regions play an important role in splicing activity.  Existing methods for detecting splicing elements are limited to finding either splice sites or enhancers and silencers, even though these elements are well-known to co-occur. We introduce SeeSite, an efficient and accurate tool for detecting splice sites and their complementary exon splicing enhancers (ESEs).  

SeeSite has three stages: graph construction, finding dense subgraphs, and recovering splice sites and ESEs along with their consensus. The third step involves solving {\em Consensus Sequence with Outliers}, an NP-complete string clustering problem. We prove that our algorithm for this problem outputs near-optimal solutions in polynomial time. 
%We encapsulate this third step as the {\em Consensus Sequence with Outliers} problem since a number of the splice junctions can be highly degenerate.  Although, this combinatorial problem is NP-complete, we use sophisticated algorithmic methods to provably solve it efficiently. 
Using SeeSite we demonstrate that ESEs are preferentially associated with weaker splice sites, and splice sites of a certain canonical form co-occur with specific ESEs. 
\end{abstract}

\section{Introduction}
 
Genes in eukaryotes typically consist of the protein-coding DNA sequences, called {\em exons}, which may be interrupted by stretches of non-coding DNA, called {\em introns}, which are spliced out of mRNA.  This RNA splicing is a fundamental biological process that is dictated by sequence markers located at {\em splice sites}, or exon-intron boundaries.  In addition to these splice sites, proximal sequences affect the splicing efficiency by recruiting helper proteins that have the effect of enhancing or silencing the splicing process \cite{Chasin}.  The classic experimental approach to discover exons and their splice sites is to map {\em Expressed Sequence Tags (ESTs)} to the reference genome \cite{EST}.  Whereas, RNA-seq is the modern approach for gene annotation with a greatly increased throughput, but has several more purposes beyond this single task \cite{RNAseq}.  However, the reduced read length and increased noise in this method yields false evidence for many non-existent splicing events, creating an opportunity for computational methods to identify novel splice sites \cite{TopHat,Trinity}.

The problem of identifying splice sites consists of two sub-problems: finding their boundaries, and characterizing their sequence markers.  A {\em marker} is a substring of a RNA molecule that is recognized by a protein.  Markers recognized by the same protein have similar nucleotide sequences.  We define a {\em motif} to be a string that represents the common nucleotide pattern recognized by the protein and thus, a motif is a representative string for a set of markers.  Splice sites can be extremely degenerate and therefore, deviate quite dramatically from their motif.  For example, mammalian 5' and 3' splice sites have the motifs (A/C)AG$\|$GT and AG$\|$GT(A/G)AGT, respectively, where '$\|$' marks the exon-intron or intron-exon boundary.  Less than 5\% of known splice sites match these canonical motifs perfectly \cite{Chasin}.  In fact, more than 60\% of the remaining splice sites have at least 3 mismatches from the consensus sequence.  In such cases where the splice site marker is weak, additional sequence markers nearby serve as binding sites for enhancer or silencer proteins and hence, are control mechanisms for splicing \cite{Barash2010}.

There exists a number of methods that address the problem of identifying splice sites.  These methods look for motif sequences exclusively at the exon-intron or intron-exon boundary.  For example, TopHat \cite{TopHat} maps as many RNA-seq reads as possible to the genome forming exon islands, then analyzes the mapping results to identify the splice sites.  MapSplice \cite{MapSplice}, HMMSplicer \cite{HMMSplicer} and SpliceMap \cite{SpliceMap} improve upon this method by using a refined alignment algorithm for the RNA-seq reads.  All these methods search for the splice sites at a specific position, which is given by the alignment of the reads.

The methods previously discussed focus only on finding splice sites and largely ignored the detection of other splicing elements, such as enhancers and silencers that occur in intron (i.e. ISE and ISS) and exon (i.e. ESE and ESS) regions.  Separate computational programs exist to detect these elements.  These methods fix the motif sequence and search for it in the vicinity of a boundary.  More specifically, they search for $\ell$-mers that are statistically enriched in a set of cases versus controls \cite{ESEfinder,RESCUE-ESE,Fedorov,Sironi,ZLC}. For example, RESCUE-ESE \cite{RESCUE-ESE} is a well-known program that compares exons with weak splice sites to those with strong splice sites reasoning that those with weak splice sites will have more enhancers. Hence, existing computational approaches for identifying splicing motifs fall into two categories based on their aim: those that detect splice sites, and those that find splicing enhancers or silencers.  

While the detection of the splice sites and their complementary splicing elements have been studied and detected separately, co-occurring relationships between them are known to exist.  For example, it has been demonstrated that evolutionary changes that weaken a splice site can be compensated by changes in the exonic splicing enhancer (ESE) or silencer (ESS) \cite{LEWE,xiao}.  Further, this relationship is also illustrated by the fact that many alternative splice sites are weaker than constitutive sites \cite{TS03,ZFG}. In addition, beyond this general relationship of having a weak splice site that is complemented by an enhancer or silencer, there exists evidence for more definite relationships between these two factors.  Xiao et al \cite{xiao} demonstrated such a relationship by showing that intronic splicing enhancers (ISEs) show specificity for different classes of splice site motifs that contribute to exon definition.

To the best of our knowledge there does not exist any well-established computational methods that detect both splice site motifs and other splicing elements, while simultaneously characterizing the relationship between them.  We introduce SeeSite, a computational program that aims at detecting splice sites and discovering splicing enhancers in exonic regions. SeeSite involves three stages: graph construction, finding dense subgraphs, and recovering splice sites and ESEs along with their consensus sequences.  The first two stages can be handled using standard methods. However,  the third stage requires solving a string clustering problem in the presence of noise since a number of the splice sites can be highly degenerate.  We formalize this clustering problem as the {\em Consensus Sequence with Outliers} problem. This combinatorial problem is NP-complete, so we have to settle for heuristic methods to solve it. On the other hand, we show that choosing the parameters to our algorithm cleverly yields strong guarantees on the quality of the output. Specifically we show that our algorithm is an efficient {\em polynomial time approximation scheme} (PTAS) unless the noise completely overwhelms the signal. This allows us to cluster the motifs into different types (corresponding to different enhancer/silencer binding sites), not all of which need to be present in all input sequences in order to identify their consensus.  We extend our theoretical findings by also giving a (less practical) PTAS for {\em Consensus Sequence with Outliers } without any restrictions on the input.

This work describes an algorithmic framework to study the problem of splice site and splicing enhancer discovery in exonic regions. The resulting method, SeeSite, is robust to uncertainty in both the sequence and position of the splice sites and ESEs.  Thus, it is suitable to finding both strong and weak canonical splice motifs and their variable proximal markers in the context of similar splice sites.  SeeSite is one of the first computational tools that aim at detecting not only splice sites but accompanying ESEs.  Hence, we believe it will be a valuable tool in elucidating splicing mechanisms in a variety of species. Our main contributions are summarized below:
\begin{itemize}
\item We develop a theoretical framework for detecting splicing elements, and characterizing the relationship between them and splice sites. 
\item We describe a polynomial-time algorithm for detecting both ESEs and splice sites.  
\item We provide empirical evidence that certain splice sites co-occur with specific ESEs 
\end{itemize}

\section{Consensus Sequence with Outliers}

%In this section, we formulate the computational problem of detecting splice junctions and ESEs in a formal way, and present an efficient algorithmic solution for this problem. 

\subsection{Problem Formulation}

We are interested in detecting splice sites that have similar sequence markers, and among those that have degenerate splice sites we aim to detect ESEs.  Given a set of possible $\ell$-mers (i.e. possible markers), we would like to determine the consensus sequence (i.e. motif), and the subset of $\ell$-mers that are the most degenerate. The following problem formally defines this task. 

\begin{definition} {\bf (Consensus Sequence with Outliers)} We denote $d(x, y)$ to be the Hamming distance between the length-$\ell$ sequences $x$ and $y$.  Given $n$ length-$\ell$ sequences $S = \{s_1, \ldots, s_n\}$ over a finite alphabet $\Sigma$ and nonnegative integer $k$, the aim of the {\em Consensus Sequence with Outliers} problem is to find a consensus sequence, $s$, and subset, $S^* \subset S$, where $n - |S^*| = k$ and $\sum_{\forall t \in S^*} d(t,s)$ is minimal.  
\end{definition}

The problem is NP-hard \cite{boucher_lo_lok}, however, it is amenable to efficient approximation algorithms that are able to work well in practice. Before defining these algorithms, we begin by giving some preliminary definitions. For a set $S$ of length-$\ell$ sequences, we denote the consensus sequence of $S$ as $c(S)$ and define it to be equal to the sequence that is obtained by picking a most-frequent character at every position with ties broken arbitrarily. We note that the tie-breaking will not affect our arguments. We denote the sum Hamming distance between a single sequence $s$ and a set of sequences $S$ as $d(S, s)=\sum_{\forall t \in S^*} d(t,s)$. The {\em Consensus Sequence With Outliers} problem can now be succinctly stated as follows: given a set $S$ of sequences and integer $k$, the objective is to find a subset $S^* \subseteq S$ of size $n^*=n-k$ such that $d(S^*, c(S^*))$ is minimized.
%Observe that the consensus sequence of set $S$, $c(S)$, minimizes $d(S, c(S))$---that is, no other sequence $x$ is closer to $S$ than $c(S)$ but some $x \neq c(S)$ could achieve $d(S,x) = d(S,c(S))$ 

Given a subset $S^* \subseteq S$ we can compute $c(S^*)$ in polynomial-time. If we are given $c(S^*)$ for the optimal solution $S^*$ (but not given $S^*$ itself) then we can recover $S^*$ from $c(S^*)$ and $S$ in polynomial-time since $S^*$ is the set of the $n-k$ sequences in $S$ that are closest to $c(S^*)$. Similarly, given any sequence $x$, we denote $S_x$ as the subset of $S$ containing the $n^*$ sequences closest to $x$. By construction $S_x$ satisfies the following inequality: $d(S', x) \geq d(S_x, x) \geq d(S_x, c(S_x))$ for any subset $S' \subseteq S$ of size $n^*$.

\subsection{Efficient Algorithms for Consensus Sequences with Outliers} \label{sec:ptas}

%Our algorithm for solving Consensus Sequence with Outliers is based on random sampling. For a given value of $\epsilon$, the algorithm selects a value for the parameter $r$ based on $\epsilon$, picks $r$ sequences $S' = (s'_1, s'_2, ... s'_r)$ from $S$ uniformly at random (with replacement), and returns the consensus sequence corresponding to $S'$. Lemma \ref{lem:min_d} shows that if $S'$ was taken from a (unknown) optimal solution $S^*$, rather than from the entire input set $S$, then in expectation $c(S')$ is almost as good as the consensus sequence for the set $S^*$.  We defer the proof of this lemma to the Appendix.  

We give a heuristic algorithm for solving {\em Consensus Sequence with Outliers} based on random sampling. 
%We then show that if the parameters of our algorithm are chosen appropriately, the algorithm will output $(1+\epsilon)$-approximate solutions in polynomial time. 
The algorithm has two parameters $r$ and $t$. It picks $r$ sequences $S' = (s'_1, s'_2, ... s'_r)$ from $S$ uniformly at random (with replacement), and finds the consensus sequence corresponding to $S'$. It repeats this process $t$ times and outputs the best consensus sequence found. The pseudocode for this algorithm is given in Algorithm~\ref{alg:rand_ptas}. 
\begin{algorithm}[ht]
%\caption{An overview of the randomized EPTAS.}
\caption{}
\label{alg:rand_ptas}
\begin{algorithmic}
%\STATE {\bf Randomized EPTAS}
\STATE {\bf Input: } $S$, $k$, $r$, $t$.
\STATE {\bf Output:} a sequence $s$ and subset, $S^* \subset S$ of size $n-k$.
%\STATE {\bf Output:} a sequence $s$ and subset, $S^* \subset S$ such that the distance between $s$ and $S^*$ is less than or equal to $(1+\epsilon)$ optimal distance.
%\STATE {\bf 1:} $\epsilon' = \frac{\epsilon}{3}$, if $\epsilon' > \frac{1}{16}$, set $\epsilon' = \frac{1}{16}$ and choose $r = \min\left(n^*,\max\left(\frac{2\ln(\frac{\sigma}{\epsilon'^2})}{\epsilon'^4}, 8\right)\right)$.
\STATE {\bf 2:} try $t$ times:
\STATE \hspace{5mm} {\bf (a):} choose a random subset of $S$ of size $r$, denoted by $S'$.
\STATE \hspace{5mm} {\bf (b):} let $S_{max}$ be a set of $k$ sequences that largest Hamming distance from $c(S')$.
\STATE \hspace{5mm} {\bf (c):} let $S^{*}$ be equal to $S$ with all the sequences in $S_{max}$ removed.
\STATE \hspace{5mm} {\bf (d):} keep track of $c(S^{*})$ with minimum $d$.
\STATE {\bf 3:} Return $c(S^{*})$ with minimum $d$ and the corresponding $S^{*}$.
\end{algorithmic}
\end{algorithm}

\subsubsection{Approximation guarantees.}
In this section we prove guarantees on the quality of the solution output by Algorithm~\ref{alg:rand_ptas}, when the parameters $r$ and $t$ are chosen appropriately. In particular we show that Algorithm~\ref{alg:rand_ptas} is an {\em efficient polynomial time approximation scheme} (EPTAS) for {\em Consensus Sequence with Outliers} if the data does not consist mainly of outliers. A {\em polynomial time approximation scheme} (PTAS) is an algorithm that for every $\epsilon > 0$ runs in polynomial time and outputs a $(1+\epsilon)$-approximate solution. Typically the running time upper bound, while polynomial for every {\em fixed} value of $\epsilon$, grows very rapidly as $\epsilon$ tends to $0$. If the exponent of the polynomial in the running time of the algorithm is independent of $\epsilon$ then the PTAS is said to be an {\em efficient PTAS} (EPTAS). To prove our bounds we prove the following technical lemma, which states that if the sample $S'$ was taken from an (unknown) optimal solution $S^*$, rather than from the entire input set $S$, then in expectation $c(S')$ is almost as good as the consensus sequence for the set $S^*$. 

\begin{lemma}\label{lem:min_d} For all $\epsilon > 0$ and $\sigma$, there exists a value of $r$ such that the following holds: if $S$ is a set of $n$ length-$\ell$ sequences over the alphabet $\Sigma$, where the size of $\Sigma$ is equal to $\sigma$, and $S'$ is a subset of $S$ of size $r$, $(s'_1, s'_2, ... s'_r)$, chosen uniformly at random, then $E[d(S, c(S'))] \leq (1+ \epsilon)d(S, c(S))$. 
\end{lemma}

\begin{proof} 
We prove that there exists a $r$ such that $E[d(S, c(S'))] \leq (1+ 2\epsilon)d(S, c(S))$. Applying this weaker inequality with $\epsilon' = \epsilon/2$ then proves the statement of the Lemma. We assume, without loss of generality, that $c(S)$ is equal to $0^{\ell}$, $\epsilon \leq 1/16$, and $r \geq 8$. We restrict interest to column $i$ of $S$, where $0 \leq i \leq \ell$, let $d_i$ be the number of nonzero symbols in column $i$ and let $z_i = n-d_i$. Observe that $d(S, c(S'))$ is equal to the sum over $i$ of the number of sequences $s \in S$ such that $s[i] \neq c(S')[i]$. By linearity of expectation it is sufficient to prove that for every $i$ we have  $E[d(S[i], c(S')[i])] \leq (1+ 2\epsilon)d_i$.
%We now consider a partricular value of $i$ and estimate the expectation of the $i$'th term of the sum above.

%% case 1
First, we assume $d_i$ is at most $\epsilon n$. Let $q$ be the probability that $c(S')[i] \neq 0$. It follows that $E[d(S[i], c(S')[i])]$ is at most $d_i(1-q) + qn$. We determine an upper bound on the probability $q$ as follows: 
\begin{eqnarray*}
q & \leq & \sum_{x = \lceil r / 2 \rceil}^r {r \choose x} \left( d_i / n \right)^x \left(1 - d_i / n \right)^{r - x} \leq  \sum_{x = \lceil r / 2 \rceil}^r 2^r \left( d_i / n \right)^x \\
& \leq & 2^r \left( d_i / n \right)^{\lceil r / 2 \rceil} \frac{1- \left( d_i / n \right)^{\lceil r / 2 \rceil}}{1 - (d_i/n)}.
\end{eqnarray*}
Since $d_i / n \leq \epsilon \leq 1/16$, we get: 
\begin{eqnarray*}
q & \leq & 2^{r + 1}  \left(d_i / n \right)^{\lceil r / 2 \rceil} \leq 2^{r + 1}  \epsilon^{\lceil r /4 \rceil} \left( d_i / n \right)^{\lfloor r / 4 \rfloor} \\
& \leq & 2^r \left( \frac{1}{16}\right)^{\lceil r / 4 \rceil} \cdot 2 \left( d_i / n \right)^{\lfloor r / 4 \rfloor} = 2 \left( d_i / n \right)^{\lfloor r / 4 \rfloor}.
\end{eqnarray*}
It follows from the last inequality, and that $r \geq 8$, that $q \leq 2 \left( d_i / n \right)^2$. Hence, we obtain the following bound on $E[d(S[i], c(S')[i])]$:

\[E[d(S[i], c(S')[i])]  \leq d_i(1 - q) + qn   
					   		\leq d_i + 2 \left( \frac{d_i}{n}\right)^2 n  
							\leq (1 + 2 \epsilon)d_i 
\]

%% case 2
Next, we assume that  $d_i > \epsilon n$.  We say that a symbol $\alpha \in \Sigma$ is a {\em good} symbol if there are at least $z_i-n\epsilon^2$ sequences in $S$ that have the symbol $\alpha$ at column $i$; any symbol that is not good is {\em bad}. If $c(S')[i]$ is a good symbol then $d(S[i], c(S')[i])$ is at most $d_i + n\epsilon^2$ and hence, is at most $(1+\epsilon)d_i$ since $d_i > \epsilon n$. 
%If $c(S')$ is a bad symbol, $E[d(c(S')[i], S[i])]$ is at most $n$. 
Let $p$ be the probability that $c(S')[i]$ is a bad symbol then, $E[d(S[i], c(S')[i])]$ is upper bounded by $(1-p)(1+\epsilon)d_i + pn$. Lastly, we determine an upper bound on $p$ to complete the proof.
% \leq (1+\epsilon + p /\epsilon) d_i$. We now show how to upper bound $p$ by $\epsilon^2$. 

Let $\alpha$ be a bad symbol and $p_\alpha$ be the probability that $c(S')[i]$ is equal to $\alpha$. We note that in order for $c(S')[i]$ to be $\alpha$, there has to be more positions equal to $\alpha$ than $0$ in $S'[i]$.  Let $X$ be the difference between the number of positions equal to $\alpha$ and the number of positions equal to $0$ in $S'[i]$.  It follows that $p_\alpha \leq \Pr[X \geq 0]$.
Let $X_j$ be an indicator variable which is $1$ if $s'_j[i]$ is equal to $\alpha$, -1 if it is equal to $0$, and $0$ otherwise.  Since $\alpha$ is a bad symbol, there are at least $\epsilon^2$ more positions equal to $0$ than positions equal to $\alpha$ in $S'[i]$ and therefore, $E[X_j] = \Pr[s'_j[i] = 0] - \Pr[s'_j[i] = \alpha ] \leq - \epsilon^2$. By linearity of expectation, we obtain $E[X] = \Sigma_{j=1}^r E[X_j] \leq - r\epsilon^2$. Using this inequality, we get $\Pr[X \geq 0] \leq \Pr[X - E[X] \geq r\epsilon^2]$. Since the $X_j$ variables are independent and difference between the upper and lower bound of $X_j$ is $2$, we can use Hoeffding's inequality to obtain the following bound.
$$\Pr[X-E[X] \geq r\epsilon^2] \leq \exp \left( \frac{-2 r^2 \epsilon^4}{r 2^2} \right)= \exp \left( \frac{r \epsilon^4}{2}\right)$$
By choosing $r= \min\left(n,\max\left(\frac{2\ln(\frac{\sigma}{\epsilon^2})}{\epsilon^4}, 8\right)\right)$, we get $p_\alpha \leq \frac{\epsilon^2}{\sigma}$. 
Finally, we bound $p$ as follows: $p \leq \sum\limits_\alpha p_\alpha \leq \sigma \frac{\epsilon^2}{\sigma} = \epsilon^2$. We can now use the upper bound on $p$ and our assumption that $d_i > \epsilon n$ to bound $E[d(S[i], c(S')[i])]$:
$$ E[d(S[i], c(S')[i])] 	  \leq (1-p)(1+\epsilon)d_i + pn  \leq (1+ \epsilon)d_i + \epsilon^2n  \leq (1+ 2\epsilon)d_i.$$
\hfill $\Box$
\end{proof}

If the number of outliers is small, then with reasonably high probability a small random subset of the input sequences will not contain any outliers. If the random sample does not contain outliers we can use Lemma~\ref{lem:min_d} to tie the quality of the output solution with the quality of the optimum solution. Based on this intuition we can prove the following theorem.
%We show that if the size of the sample and the number of repetitions of the experiment are chosen appropriately then there exists a good bound on the quality of the output of this natural heuristic.  
%For inputs where the noise does not completely overwhelm the data, i.e.~when $k \leq cn$ for $c < 1$, the dependence on the running time of our approximation scheme for Consensus Sequence with Outliers is good; more specifically, it is an EPTAS.

%Using Lemma~\ref{lem:min_d} we can show the following performance guarantee
%
%
%\begin{proof}
%It follows from Lemma~\ref{lem:min_d} that there exists an integer $r$ such that $$E[d(S^*, c(S'))] \leq (1+\epsilon)d(S^*, c(S^*))$$ if $S'$, the set of $r$ of sequences chosen from $S$, is from an (unknown) optimal solution $S^*$. Some subset $S'$ of $S^*$ must achieve expectation. The algorithm guesses this set $S'$ by trying all possible $n^r$ subset of $S$ of size $r$. Let $x=c(S')$. The algorithm returns the set $S_x$ of the $n^*$ sequences closest to $x$. This set satisfies $d(S_x, c(S_x)) \leq d(S_x, x) \leq d(S^*, x) \leq (1+\epsilon)d(S^*,c(S^*))$, concluding the proof.  \hfill $\Box$ 
%\end{proof}
%
%
%
\begin{theorem} \label{lem:min_d_rand} There exists a randomized EPTAS for {\em Consensus Sequence with Outliers} for inputs when $k \leq cn$ for $c < 1$. The algorithm runs in time $\frac{1}{(1-c)^r} \cdot f(\epsilon)(n\ell)^{O(1)}$ and outputs a $(1+\epsilon)$-approximate solution with probability $1/2$.
\end{theorem}
\begin{proof} 
The algorithm selects a value for $r$ such that for a random subset $S'$ of the unknown optimal solution $S^*$ the inequality $E[d(S^*, c(S'))] \leq (1+\frac{\epsilon}{3})d(S^*, c(S^*))$ holds. It follows from Lemma~\ref{lem:min_d} that this can be done so that $r$ only depends on $\epsilon$.
%Specifically it sets $\epsilon' = \min\left(\frac{\epsilon}{3},\frac{1}{16}\right)$ and $r = \min\left(n^*,\max\left(\frac{2\ln(\frac{\sigma}{\epsilon'^2})}{\epsilon'^4}, 8\right)\right)$. 
We show that a single iteration of the outer loop of Algorithm~\ref{alg:rand_ptas} with this choice for $r$ yields a $(1+\epsilon)$-approximate solution with probability $(1-c)^r \cdot f(\epsilon)$. Then setting $t = O\left(\frac{1}{(1-c)^r \cdot f(\epsilon)} \right)$ yields the statement of the theorem. 

It remains to find a sufficient lower bound of the probability that the set returned by a single iteration of the outer loop of Algorithm~\ref{alg:rand_ptas} is a $(1+\epsilon)$-approximation. Since $k \leq cn$, it follows that the probability that $S'$ is taken from an (unknown) optimal solution $S^*$ is at least $(\frac{n-cn}{n})^r = (1-c)^r$.  If $S'$ is taken from $S^*$ then by Lemma~\ref{lem:min_d} we have that $E[d(S^*, c(S'))] \leq (1+\frac{\epsilon}{3})d(S^*, c(S^*))$. By Markov's inequality \cite[p. 311]{GR} the probability that $d(S^*, c(S'))$ exceeds expectation by a factor at least $1+\frac{\epsilon}{3}$ is at most $\frac{1}{1+\frac{\epsilon}{3}}$. Hence, with probability $f(\epsilon)$ for some function $f$ of $\epsilon$ we have that: $$d(S^*, c(S')) \leq \left( 1+\frac{\epsilon}{3} \right) d(S^*, c(S^*)) \cdot \left( 1+\frac{\epsilon}{3} \right),$$ which is at most $(1+\epsilon)d(S^*, c(S^*))$ when $\left( \frac{\epsilon}{3} \right)^2 \leq \frac{\epsilon}{3}$. 
In particular, this holds if $\epsilon \leq 3$, concluding the proof. 
\hfill $\Box$ \end{proof}
We note that one would expect natural inputs to contain substantially fewer outliers than $n/2$, and that Markov's inequality is a very pessimistic bound for the probability of achieving expectation. Hence, it is likely that for reasonable inputs Algorithm~\ref{alg:rand_ptas} performs much better in practise than the proved bounds. In fact, on our synthetic data the algorithm vastly outperformed the theoretical bounds. 

Using Lemma~\ref{lem:min_d} we can also get a simple deterministic PTAS (but not an EPTAS) for {\em Consensus Sequence with Outliers} without any assumptions on the relationship between $k$ and $n$. Specifically we prove the following theorem.
%Theorem \ref{lem:min_d_det}) shows that this algorithm has the characteristic of being a {\em polynomial-time approximation scheme} (PTAS).  A PTAS for a minimization problem is an algorithm which takes an instance of the problem and a parameter $\epsilon > 0$ and, in polynomial-time, produces a solution that is within a factor $1 + \epsilon$ of being optimal. We defer the proof of this theorem to the Appendix. 
\begin{theorem} \label{lem:min_d_det} There exists a PTAS for {\em Consensus Sequence with Outliers}. \end{theorem}

\begin{proof}
It follows from Lemma~\ref{lem:min_d} that there exists an integer $r$ such that if $S'$, the set of $r$ of sequences chosen from $S$, is from an (unknown) optimal solution $S^*$ then
$E[d(S^*, c(S'))] \leq (1+\epsilon)d(S^*, c(S^*))$. Some subset $S'$ of $S^*$ must achieve expectation. The algorithm guesses this set $S'$ by trying all possible $n^r$ subset of $S$ of size $r$. Let $x=c(S')$. The algorithm returns the set $S_x$ of the $n^*$ sequences closest to $x$. This set satisfies $d(S_x, c(S_x)) \leq d(S_x, x) \leq d(S^*, x) \leq (1+\epsilon)d(S^*,c(S^*))$, concluding the proof. \hfill $\Box$
\end{proof}

\section{System and Methods}

The basic algorithm behind the SeeSite software goes through the following phases: graph construction, identification of dense subgraphs, and recovery of splicing sites or ESEs.  In this section we will discuss each of these steps in greater detail.  The inputs to SeeSite are $m$ the minimum subgraph size, the maximum number of total mismatches $d$, the number of outlier sequences $k$, and $b$ the maximum number of mismatches for the existence of an edge.  In addition, there is the option to restrict the search to ESEs or splice junctions that have a specific canonical form.    

\subsubsection{Graph Construction}\label{graph_construction} 
We construct a graph from the data, with each vertex representing an $\ell$-mer. The goal of the construction is to ensure that dense subgraphs correspond to closely related subsequences. The dense subgraphs will then be passed on to later stages of the algorithm for further consideration. For the remainder of this section we assume that all input sequences have length at most $L$. We now give a formal definition of the constructed graph.

\begin{enumerate}
\item The vertex set contains a vertex $v_{i,j}$ representing the $l$-length subsequence in sequence $i$ starting at position $j$, for each $i$ and $j = 1, 2, \dots, L-l+1$. There are at most $n(L-l+1)$ vertices.
\item Each pair of vertices $v_{i,j}$ and $v_{i',j'}$, for $i \neq i'$ is joined by an edge when the Hamming distance between the two represented subsequences is at most $b$.
\end{enumerate}

\noindent This graph is represented by a symmetric adjacency matrix, where each entry is 0 for a non-edge, or a positive weight for an edge. We reduce the running time of searching $G$ by considering subgraphs of $G$, $\{G_0, G_1, \dots, G_{L-1}\}$, where $G_i$ is the subgraph induced by vertex a reference vertex, denoted as $v_{R,i}$, and its neighbors (for some arbitrary choice of reference sequence $R$).  We note that similar graph constructions have been used by Yang and Rajapakse \cite{YR04}, Pevzner and Sze \cite{PS00}, and Yang et al. \cite{YDA10}.

\subsubsection{Detection of Dense Subgraphs}

We implemented a modified version of the MCQ algorithm of Tomita and Seki \cite{TM03} to enumerate all dense subgraphs of size at least $m$. Each dense subgraph represents a motif instance. We chose MCQ due to the experimental work showing that when compared with other existing algorithms, it is the most efficient and practical for dealing with large graphs \cite{TAM}. The underlying idea of this branch-and-bound, depth-first search method is to begin with a small dense subgraph (i.e. single vertex), add a vertex to it if and only if it is connected to some minimum number of vertices already contained in the subgraph, and halt when no other vertices can be added.  If the subgraph has size at least $m$ then it is returned.  Each of the vertex sets outputted by the the MCQ algorithm corresponds to a set of $\ell$-mers that needed to be considered further in the next stage of the algorithm.

\subsubsection{Recovery of Splice Junctions and ESEs}

Each of set of $\ell$-mers identified in the previous step represents an instance of the {\em Consensus Sequence with Outliers} problem. Hence, Algorithm \ref{alg:rand_ptas} is used to distinguish between sets of $\ell$-mers that represent splice junctions or ESEs, or spurious patterns in the input data.  We note that the size of the set of $\ell$-mers (parameter $n$ in {\em Consensus Sequence with Outliers} problem formulation) is at least $m$.  The output is the set of all valid splice junctions, their consensus sequence, and outliers, or the set of all valid ESEs.  We will see in the next subsection how this step can be adapted to detect both these splicing elements.

\subsection{Detecting Co-occurring Splice Sites and ESEs}

We run SeeSite in a two-fold manner in order to detect the splice junctions and ESEs associated with weak splice sites. First, SeeSite is ran to detect all possible splice junctions. In this first run of SeeSite, Algorithm \ref{alg:rand_ptas} is used (with appropriate values of $r$ and $t$) to determine the set of $\ell$-mers that correspond to sets of splice junctions; the consensus sequence and outliers for each set of splice junctions are also outputted at this stage. The outliers correspond to splice junctions that are highly degenerate. The exon regions corresponding to a set of outliers are then input into the second run of SeeSite. In this second run, the ESEs are detected in each of these groups of exons. Running SeeSite in this two-fold manner allows us to detect both strong and weak splice junctions along with the ESEs associated with weak splice junctions. 

As previously mentioned, the {\em Consensus Sequence with Outliers} problem formulation can be adapted to tailor the search for splice sites or ESEs. When searching for splice sites in the first run of SeeSite, we set $k$, the number of outliers, equal to $\lfloor n/4 \rfloor$, where $n$ is the number of $\ell$-mers in the motif instance. When searching for ESEs in the second run of SeeSite, the number of outliers is equal to zero, that is Algorithm \ref{alg:rand_ptas} will return the majority sequence.

\section{Investigation of Splice Sites in the Human Genome}

The human genome has many multi-exon genes with excellent EST coverage and high-quality annotation.  Thus, it provides a good source of known splice sites on which we can evaluate SeeSite's ability to detect canonical and non-canonical motifs as well as their complementary ESEs. Our benchmark dataset consists of 10,000 known splice sites from the human genome (hg19 assembly) and its reference annotation (RefSeq).  To capture ESEs, we extracted two 100bp subsequences centered at the 5' and 3' splice sites flanking each known intron.

\subsection{Dectection of Canonical and Non-Canonical Splice Sites}

We ran SeeSite with the minimum size of the subgraph (parameter $m$) equal to 100, and varied the values of $\ell$, $d$ and $b$.  The parameter $n$, which is the number of $\ell$-mers that need to be considered in the third step of the algorithm, is at least $m$, and the maximum number of outlier sequences (parameter $k$) equal to $\lfloor n/4 \rfloor$. For each set of splice sites corresponding to one consensus sequence, there exists a (possibly empty) set of outlier sequences.  SeeSite was capable of detecting splice sites in 9,208 of the 10,000 genes considered, 87\% of these sites overlapped with known gene models that have been verified by ESTs.  

One of the main advantages of SeeSite is the accuracy in detecting weak splice sites.  A metric, referred to as the {\em consensus value} (CV), is used to gauge the degeneracy (or strength) of the splice site and is a index that ranges from 100 (perfect consensus) to 0 (worst consensus) \cite{SSH90,ZLC}.  Of the 9,208 sites detected by SeeSite, less than 5\% had a CV greater than 80, and more than two thirds had a CV less than 77.   The splice sites that were identified as ``outliers'' by SeeSite had a mean CV of 68, with the variance of the distribution of the CVs being 2.5.  

SeeSite detected all splice sites with well-known canonical forms, as well as, identified non-canonical sites. For example, the consensus splice site pattern that has GT(G/A)AGT for the first 6bp of the intron and (T/C)AG as the last 3bp of the intron is well-known canonical splicing form in {\em Homo sapien} data.  Although these splice sites are highly degenerate and leave GT-AG as the only reliable splice pattern in the {\em Homo sapien} data \cite{Stamm}, SeeSite, TopHat \cite{TopHat} and HMMSplicer \cite{HMMSplicer} were capable of detecting majority of these sites. SeeSite identified the GT-AG, GC-AG, and AT-AC splicing site patterns in 87\%, 3\%, and 0.5\% of the genes considered.  Of the GT-AG splicing sites identified, 2\% matched perfectly to the consensus (CV of 100), 4\% had a CV between 90 and 80, 26\% had a CV between 80 and 77, and the remaining 68\% had a CV less than 77.  

\begin{table}
\begin{center}
\begin{tabular}{p{5cm}|p{5cm}|p{5cm}}
\hline
Splicing Pattern 													& ESEs with High Correlation								&	ESEs with Low Correlation  \\ 
\hline
\hline
\multirow{3}{*}{\tt GT(G/A)AGT$\|$(T/C)AG} 		&{\tt AAAAGA}	(24\%) 											& {\tt ATGGCG} \\
 																			&{\tt AAGAAT}	(21\%) 											& {\tt CAAGAT} \\
 																			&{\tt GAAAAT}	(17\%) 											& {\tt ATGAGA}	 \\
 																			&{\tt TGGAAA}	(16\%) 											& {\tt ATGAGA}	 \\
\hline																		
\multirow{3}{*}{\tt GCA(T/A)G(G/T)$\|$(T/A)AC} &{\tt ATGAGA}	(33\%) 											& {\tt AAAAGA} \\
 																			&{\tt TACAGA}	(34\%) 											& {\tt GAAAAT} \\
 																			& 																			& {\tt ATGGAA} \\
 																			& 																			& {\tt CAAGAT} \\
\hline																		
\multirow{3}{*}{\tt ATG(C/A)T(G/A)$\|$(G/T)AC} &{\tt ATGGAA}	(41\%) 											& {\tt AAAAGA} \\
 																			&{\tt CAAGAT} (23\%) 											& {\tt GAAAAT} \\
 																			& 																			& {\tt TACAGA} \\
 																			& 																			& {\tt ATGAGA} \\

\hline
\end{tabular}
\end{center}
\caption{Some examples that illustrate the relationship between co-occuring splice sites and ESEs that control splicing expression. We searched for all possible ESEs in over 3000 weak splice sites of the form GT-AG, over 500 weak splice sites of the form GC-AC, and over 100 weak splice sites of the form AT-AC. We witnessed the existence of a number of ESEs that are either likely or unlikely to occur in the presence of splice sites that have a specific pattern.  For each splice site pattern we listed the ESEs with high correlation and low correlation.  Those with low correlation occurred in less than 10\% of the exons associated with the corresponding splicing pattern.  The percentage of occurrence is given in brackets for each of the ESEs with high correlation. }
\label{tab:eses}
\end{table}

\subsection{ESEs are Paired with Weaker Splice Sites}

Whereas most past computational methods search for either splicing site sites or ESEs, we sought pairs of patterns that demonstrate an unusually strong tendency to co-occur across exons.  In order to accomplish this task, we ran SeeSite on each set of outlier sequences with minimum subgraph size (parameter $m$) equal to $50$, varied values of $d$ and $b$, $\ell$ equal to $5$ and $6$, and $k=0$. The largest set of outlier sequences was $198$.  We validated our findings by comparing the ESEs found by SeeSite with those identified by RESCUE-ESE \cite{RESCUE-ESE}\footnote{The list of identified ESEs found by using RESCUE-ESE are here: {\tt http://genes.mit.edu/burgelab/rescue-ese/ESE.txt}.}.  Of the sequences found at this stage by SeeSite, 95\% were identified by Fairbrother et al \cite{RESCUE-ESE} as being an ESE.  A summary of our results are found in Table \ref{tab:eses}, which shows a pairing of ESEs and specific splice site patterns.  Our results give strong evidence toward the existence of co-occurring splicing elements---that is, a pairing of splice sites with specific ESEs.   For the weak splice sites of each form, GT-AG, GC-AC, and AT-AC, we determined the ESEs that are occur most frequently and those that occur the most infrequently and compared these across different splicing sites.  We witnessed that the occurrence of several ESEs have high correlation to the corresponding splicing site having a specific pattern (i.e.  GT-AG, GC-AC, and AT-AC), while others have a low correlation.  Hence, there exists strong evidence for a pairing of ESEs to splicing sites.  

In addition, we considered the number of exons that are associated with strong splice sites sites and that are paired with an ESE, with the number of exons that are associated with weak splice sites sites and that are paired with an ESE.  We found that 90\% of weak splice sites are paired with an ESE, as opposed to 30\% of strong splice sites.  In this context, we refer to a splice site as strong if the CV is greater than or equal to 85. This statistic supports the ongoing conjecture that co-occurring pairs are contributing to splicing by compensating for a lack of strong splicing signals. 

\section{Conclusion and Future Work}

SeeSite is a computational tool for detecting splice sites and ESEs, and identifying co-occuring relationships between these sites. Our results suggest the existence of several non-canonical splice site patterns and demonstrates a possible synergistic relationship between ESEs and different classes of splice site patterns.  Future experimental work is needed to resolve whether these relationships between splicing elements, and non-canonical splice sites are biologically significant or simply spurious correlations detected in the data. However, we believe this is unlikely given the strength and frequency of a number of the patterns.  Determining the exact biological relationship of these paired splicing elements, i.e. they could act negatively to promote exon skipping, whether these pairs are equally conserved over evolution warrants further investigation.

\end{document}